\documentclass[]{article}
\usepackage{amsmath,amssymb,amsthm,graphicx}
\usepackage{fullpage}

\newtheorem{theorem}{Theorem}[]
\newtheorem{lemma}[theorem]{Lemma}

\theoremstyle{definition}

\usepackage{latex_abbreviations}
\usepackage{mathtools}
\usepackage{algorithm}
\usepackage[noend]{algorithmic}
\usepackage[square,numbers]{natbib}
\usepackage{hyperref}
\usepackage{cleveref}
\allowdisplaybreaks

\newcommand{\argmax}{\mathrm{argmax}}

\newcommand{\bfone}{{\bf 1}}

\title{An improved algorithm for the submodular secretary\\problem with a cardinality constraint} 
\author{Kaito Fujii\\University of Tokyo\\\href{mailto:kaito\_fujii@mist.i.u-tokyo.ac.jp}{\nolinkurl{kaito\_fujii@mist.i.u-tokyo.ac.jp}}}

\begin{document}

\maketitle

\begin{abstract}
We study the submodular secretary problem with a cardinality constraint.
In this problem, $n$ candidates for secretaries appear sequentially in random order.
At the arrival of each candidate, a decision maker must irrevocably decide whether to hire him.
The decision maker aims to hire at most $k$ candidates that maximize a non-negative submodular set function.
We propose an $(\rme - 1)^2 / (\rme^2 (1 + \rme))$-competitive algorithm for this problem, which improves the best one known so far.
\end{abstract}



\section{Introduction}
In the classical secretary problem, a decision maker aims to hire the best one out of $n$ candidates.
Each candidate has a real value that expresses his skill.
The decision maker knows only the number $n$ of candidates in the beginning.
The candidates appear in random order one by one, and at the arrival of each candidate, his value is revealed.
Just after observing the value, the decision maker must decide whether to hire this candidate or not.
This decision is irrevocable, and the decision maker can hire only one candidate.
The goal of this problem is to hire the candidate with the largest value with probability as high as possible.
An asymptotically optimal strategy for this problem was stated by \citet{Dynkin63}.
This strategy ignores the first $\lfloor n / \rme \rfloor$ candidates without hiring and hires the first candidate better than these first $\lfloor n / \rme \rfloor$ candidates.
This algorithm hires the best one with probability at least $1 / \rme$~\citep{Dynkin63}.

A multiple-choice variant was proposed by \citet{Kleinberg05}.
In contrast to the classical secretary problem, $k$ candidates can be selected in this variant, where $k \in \bbZ_{> 0}$ is the maximum number of hired candidates.
The decision maker knows $n$ and $k$ in advance.
Let $V$ be the set of all candidates and $w \colon V \to \bbR_{\ge 0}$ a non-negative weight of each candidate.
At each arrival of candidates, the decision maker must decide whether to hire the candidate or not according to the revealed value of the candidate.
The objective is to maximize the sum of values of the hired candidates, that is, $\sum_{v \in S} w(v)$, where $S \subseteq V$ is the set of the hired candidates.

As a further extension of the multiple-choice secretary problem, \citet{BHZ13} and \citet{GRST10} proposed \textit{submodular secretary problems}, in which the objective function $f \colon 2^V \to \bbR_{\ge 0}$ is not the sum of values of the hired candidates, but a submodular set function.
A set function $f \colon 2^V \to \bbR_{\ge 0}$ is called submodular if it satisfies $f(S \cup \{v\}) - f(S) \ge f(T \cup \{v\}) - f(T)$ for all $S \subseteq T \subseteq V$ and $v \in V \setminus T$.
We assume that for any subset $S$ of candidates that already appeared, the value of $f(S)$ can be computed in $\rmO(1)$ time.
Among various constraints for submodular secretary problems, in this study we focus on \textit{cardinality constraints}, under which the decision maker can hire at most $k$ candidates out of $n$ candidates.

Submodular secretary problems have many applications.
A key application is \textit{online auctions}~\citep{Bateni16}, in which the seller decides who gets products in an online fashion out of buyers declaring their bids one by one.
One of the goals is to design a mechanism that maximizes the utility of the agents.
Submodularity of utility functions is often assumed due to its equivalence to the property of diminishing returns.
Also, submodular secretary problems are applied to machine learning tasks such as stream-based active learning~\citep{FK16} and the interpretation of neural networks~\citep{Elenberg2017}.

The quality of an algorithm for submodular secretary problems is evaluated by its \textit{competitive ratio}.
Let $\calA(f, \sigma) \subseteq V$ be the (possibly randomized) output of algorithm $\calA$, where $f \colon 2^V \to \bbR_{\ge 0}$ is the objective function and $\sigma \in \Sigma_n$ is a permutation of $n$ candidates.
The competitive ratio $\alpha(\calA)$ of an algorithm $\calA$ is defined as
\begin{equation*}
	\alpha(\calA) = \inf_{f} \frac{\bbE [f(\calA(f, \sigma))]}{\max_{S^* \subseteq V} f(S^*)},
\end{equation*}
where the expectation is taken over the random permutation and random factors in the algorithm.

\subsection{Related work}
For the multiple-choice secretary problem, that is, the case where $f$ is a linear function, $(1 - \rmO(1 / \sqrt{k}))$-competitive~\citep{Kleinberg05} and $(1 / \rme)$-competitive~\citep{BIKK07} algorithms were proposed.

For the submodular secretary problem with a cardinality constraint, \citet{GRST10} proposed a $1/1417$-competitive algorithm and \citet{BHZ13} proposed a $1/(8\rme^2)$-competitive algorithm.
Several studies of submodular secretary problems assume the \textit{monotonicity} of the objective function, which is defined as the condition that $f(S) \le f(T)$ for all $S \subseteq T \subseteq V$.
Under the monotonicity assumption, \citet{BHZ13} proposed a $0.090$-competitive algorithm and \citet{FNS11} showed this algorithm achieves $0.170$-approximation.
The best one for the monotone case is $(1 - \rmO(1/\sqrt{k})) / \rme$-approximation with exponential running time by \citet{KT17}.

Submodular secretary problems under more general constraints have been considered in several studies such as \citet{MTW16} and \citet{FZ18}.
Secretary problems with a set function lacking submodularity also have been studied, such as monotone subadditive functions by \citet{RS17} and monotone functions with bounded supermodular degree by \citet{FI17}.

Submodular maximization with a cardinality constraint in the offline setting has been studied extensively.
In the offline setting, the ground set is given in advance and the objective value of any subset can be obtained from the beginning.
The best algorithm known so far achieves $0.385$-approximation by \citet{BF16}.
Under the monotonicity assumption, the greedy algorithm is known to achieve $(1 - 1 / \rme)$-approximation~\citep{NWF78}.


\subsection{Our result}
In this study, we propose an algorithm with an improved competitive ratio for the submodular secretary problem with a cardinality constraint.
The objective function is assumed to be non-negative and submodular, but not necessarily monotone.
Our algorithm is based on the one proposed for the monotone case by \citet{BHZ13}.
We slightly modify their algorithm so that it works for the non-monotone case.
In the analysis of the competitive ratio, we use a lemma proved by \citet{BFNS14}, which was originally used for the offline setting.
While \citet{BFNS14} designed a randomized algorithm and applied this lemma to the analysis of its approximation ratio, we utilize a random factor of the ordering of candidates.
The resulting algorithm is $(\rme - 1)^2 / (\rme^2 (1 + \rme)) \approx 0.107$-competitive.
To the best of our knowledge, this competitive ratio improves the best one known so far, which is $1 / (8\rme^2) \approx 0.0169$~\citep{BHZ13}.

\section{Preliminaries}
Let $V$ be a finite set of size $n$.
The elements of $V$ arrive one by one in random order, i.e., the order of $V$ is chosen uniformly at random out of $n!$ permutations.
The non-negative submodular objective function $f \colon 2^V \to \bbR_{\ge 0}$ can be accessed through an value oracle, which returns the value of $f(S)$ in $\rmO(1)$ time for any subset $S$ of those who already arrived.
At the arrival of element $v \in V$, the algorithm must decide whether to add it to the solution or reject it irrevocably.
The algorithm can select at most $k$ elements.
We define the marginal gain of adding an element $v \in V$ to the current solution $S \subseteq V$ as $f(v|S) \coloneqq f(S \cup \{v\}) - f(S)$.

A continuous-time model~\citep{FNS11} is a problem setting equivalent to the random order model.
In this model, each element is assigned to an \textit{arrival time} that is generated from the uniform distribution on $[0, 1]$, and the decision maker can observe each element at its arrival time.
We can transform a problem instance in the random order model to one in the continuous-time model.
Suppose we generate $n$ random numbers from the uniform distribution on $[0,1]$ and sort them in ascending order in advance.
By assigning these values to each element in order of arrival, we can transform the random order model to the continuous-time model.

\section{The modified classical secretary algorithm}
In this section, we describe a modified version of the classical secretary algorithm, which is used as a subroutine of our proposed method.
The modified classical secretary algorithm is based on the continuous-time version of the classical secretary problem~\citep{FNS11}, which ignores all elements that arrive before time $1 / \rme$ and selects the first element with value higher than all ignored elements.
They proved that their method selects the best one with probability at least $1/\rme$.
However, we need another property: The probability that the algorithm selects each element must be upper-bounded.
Their method does not satisfy this property if $n$ is small.
Thus we make a small modification to this algorithm, and obtain the one satisfying this property.
This modification is applied only when there exists no element that arrives before time $1/\rme$.
In this case, the modified algorithm selects the first element with probability $1/(\rme t_1)$, where $t_1$ is the time when the first element arrives, while the original one always selects it.
This modified algorithm is described in \Cref{alg:classical}.
\begin{algorithm}[t]
	\caption{The modified classical secretary algorithm}
	\begin{algorithmic}[1]\label{alg:classical}
	\REQUIRE A randomly ordered elements $v_1,\cdots,v_n$.
	\STATE Generate $n$ numbers from the uniform distribution on $[0,1]$ and sort them in ascending order. Let $(t_i)_{i=1}^n$ be these sorted numbers.
	\STATE Regard $t_i$ as the arrival time of the $i$th element $v_i$.
	\STATE Ignore all elements that arrive before time $1/\rme$.
	\IF{there is no element that arrives before time $1/\rme$}
		\STATE Select the first element with probability $1/(\rme t_1)$. Otherwise, no element is selected.
	\ELSE
		\STATE Let $\theta$ be the largest value of the elements that arrive before time $1/\rme$.
		\STATE Select the first element with value at least $\theta$.
	\ENDIF
\end{algorithmic}
\end{algorithm}

\begin{lemma}\label{lem:classical}
	\Cref{alg:classical} is $1/\rme$-competitive and for each element $v \in V$, the probability that \Cref{alg:classical} selects $v$ is at most $1/\rme$.
\end{lemma}

\begin{proof}
	Let $v^* \in V$ be an element with the largest value.
	Suppose $v^*$ appears after time $1/\rme$ and its arrival time is $t^* \in [1/\rme,1]$.
	Below, we fix the arrival time $t^*$ and consider the event that $v^*$ is selected by the algorithm.
	This event happens if one of the following two conditions holds.
	\begin{itemize}
		\item $v^*$ is the first element and the algorithm decides to select $v^*$.
		\item $v^*$ is not the first element, and $\hat{v}$ appears before time $1/\rme$, where $\hat{v}$ is the best element among those who appear before time $t^*$.
	\end{itemize}

	The probability that $v^*$ is the first element is equal to the probability that all other elements arrive after time $t^*$, then it is $(1 - t^*)^{n-1}$.
	Since the algorithm selects the first element with probability $1/(\rme t_1)$, the first condition holds with probability $(1-t^*)^{n-1} / (\rme t^*)$.

	To consider the second condition, we fix the set of elements $S \subseteq V$ that arrive before time $t^*$.
	Then the largest element among them, $\hat{v} \in \argmax_{v \in S} w(v)$, is also determined.
	Under this condition, the probability that $\hat{v}$ arrives before time $1/\rme$ is $1/(\rme t^*)$ since the arrival time of $\hat{v}$ conforms to the uniform distribution on $[0, t^*]$.
	Since this holds for any $S$, the second condition holds with probability $\{1-(1-t^*)^{n-1}\} / (\rme t^*)$.
	In total, the probability that the algorithm selects $v^*$ is 
	\begin{align*}
		&\int_{1/\rme}^1 \frac{(1-t^*)^{n-1}}{\rme t^*} + \frac{1-(1-t^*)^{n-1}}{\rme t^*} \rmd t^*\\
		&= \int_{1/\rme}^1 \frac{1}{\rme t^*} \rmd t^*\\
		&= \frac{1}{\rme}
	\end{align*}
	Then this algorithm is $1/\rme$-competitive.
	
	Let $v \in V$ be any element and $t$ be its arrival time.
	If $v$ is selected, one of the following two conditions holds.
	\begin{itemize}
		\item $v$ is the first element and the algorithm decides to select $v$.
		\item $v$ is not the first element, $\hat{v}$ appears before time $1/\rme$, and the value of $\hat{v}$ is less than that of $v$, where $\hat{v}$ is the best element among those who appear before time $t$.
	\end{itemize}
	Since the second condition is stronger than the second condition for the optimal element $v^*$, the probability that one of these conditions is satisfied is less than that of $v^*$.
	Therefore, the probability that any element is selected is at most $1/\rme$.
\end{proof}

\section{Algorithm and analysis}
In this section, we illustrate our proposed method and provide an analysis of its competitive ratio.

Our proposed method utilizes the idea of the continuous-time model.
To transform the random order model to the continuous-time model, we generate $n$ random numbers from the uniform distribution on $[0, 1]$ and sort them in ascending order in advance.
Let $t_1 \le \cdots \le t_n$ be these numbers.
By regarding $t_i$ as the arrival time of the $i$th element, we can obtain an instance of the continuous-time model.
We partition the sequence of the elements into $k$ segments $V_1,\cdots,V_k$ by separating $[0, 1]$ into equal-length time windows.
Then we apply \Cref{alg:classical} to each segment with regarding the marginal gain $f(v|S_{l-1})$ as the weight for element $v \in V_l$, where $S_{l-1}$ is the solution just before segment $V_l$.
If \Cref{alg:classical} selects an element $s_l \in V_l$ and its marginal gain is non-negative, we add it to the solution.
The detailed description of the algorithm is provided in \Cref{alg:proposed}.

\renewcommand{\algorithmicrequire}{\textbf{Input:}}
\begin{algorithm}[t]
	\caption{Our proposed method}
	\begin{algorithmic}[1]\label{alg:proposed}
	\REQUIRE A randomly ordered elements $v_1,\cdots,v_n$.
	\STATE Generate $n$ numbers from the uniform distribution on $[0,1]$ and sort them in ascending order. Let $(t_i)_{i=1}^n$ be these sorted numbers. Regard $t_i$ as the arrival time of the $i$th element $v_i$.
		\STATE Partition $V$ into $k$ segments $V_1,\cdots,V_k$, where $V_l = \{v_i \in V \mid t_i \in [(l-1)/k, l/k] \}$ is the set of elements that arrive from time $(l-1)/k$ to time $l/k$ for each $l \in [k]$.
	\STATE Let $S_0 \coloneqq \emptyset$.
	\FOR{each $l = 1,\cdots,k$}
		\STATE Apply \Cref{alg:classical} to $V_l$ with weight $f(v|S_{l-1})$ for each $v \in V_l$.
		\IF{\Cref{alg:classical} selects an element $s_l \in V_l$ and $f(s_l|S_{l-1}) \ge 0$}
			\STATE $S_{l} \gets S_{l-1} \cup \{ s_l \}$.
		\ELSE
			\STATE $S_{l} \gets S_{l-1}$.
		\ENDIF
	\ENDFOR
	\STATE \textbf{return} $S_k$.
\end{algorithmic}
\end{algorithm}

In the proof, we utilize the following lemmas.
\begin{lemma}[{Lemma 2.2 of~\citep{FMV11}}]\label{lem:useful1}
    Let $f \colon 2^V \to \bbR$ be submodular. Denote by $A(p)$ a random subset of $A$ where each element appears with probability $p$ (not necessarily independently). Then, $\bbE[f(A(p))] \ge (1 - p) f(\emptyset) + p \cdot f(A)$.
\end{lemma}

\begin{lemma}[{Lemma 2.2 of~\citep{BFNS14}}]\label{lem:useful2}
    Let $f \colon 2^V \to \bbR$ be submodular. Denote by $A(p)$ a random subset of $A$ where each element appears with probability at most $p$ (not necessarily independently). Then, $\bbE[f(A(p))] \ge (1 - p) f(\emptyset)$.
\end{lemma}

\begin{theorem}
	\Cref{alg:proposed} is $(\rme - 1)^2 / (\rme^2 (1 + \rme))$-competitive for any non-negative submodular objective function and cardinality constraint.
\end{theorem}

\begin{proof}
	Let $S^* \in \argmax_{S \colon |S| \le k} f(S)$ be an optimal solution.
	Let $U$ be a set with the largest objective value that has at most one element in each partition, i.e., $U \in \argmax\{f(U) \mid {U \subseteq S^*, ~ \forall l \in [k] \colon |U \cap V_l| \le 1}\}$.
	Let $\tilde{U}$ be a random subset of $S^*$ obtained by selecting an element uniformly at random from each non-empty $S^* \cap V_l$ where $l \in [k]$.
	From the definition, we have $\bbE[f(U)] \ge \bbE[f(\tilde{U})]$.

	From the independence between arrival times of elements, each element of $S^*$ is included in partition $V_l$ with probability $1/k$ for each $l \in [k]$.
	Then we have
	\begin{equation*}
		\Pr(V_l \cap S^* \neq \emptyset) = 1 - \left(1 - \frac{1}{k}\right)^{|S^*|}
	\end{equation*}
	for each $l \in [k]$.
	From the linearity of expectation, we have
	\begin{align*}
		\bbE[|\tilde{U}|]
		&= \bbE\left[\sum_{l = 1}^k \bfone_{V_l \cap S^* \neq \emptyset}\right]\\
		&= \sum_{l = 1}^k \Pr (V_l \cap S^* \neq \emptyset) \\
		&= k\left(1 - \left(1 - \frac{1}{k}\right)^{|S^*|}\right)\\
		&\ge |S^*| \left(1 - \frac{1}{\rme} \right).
	\end{align*}
	By considering the randomness of $V_l$ for all $l \in [k]$, we can see that each element in $S^*$ is included in $\tilde{U}$ with the same probability, then we can define $p \coloneqq \Pr(v \in \tilde{U})$ for all $v \in S^*$.
	Since $\bbE[|\tilde{U}|] = p |S^*|$, we have $p \ge 1 - 1 / \rme$.
	By applying Lemma \ref{lem:useful1}, we obtain
    \begin{align}
		\bbE[f(U)]
		&\ge \bbE[f(\tilde{U})]\nonumber\\
		&\ge (1 - p) f(\emptyset) + p f(S^*)\nonumber\\
		&\ge (1 - 1 / \rme) f(S^*).\label{eq:u}
    \end{align}

	Fix partitions $V_1,\cdots,V_k$ and $U$.
	We consider the randomness of the ordering of each partition.
	Since \Cref{alg:classical} selects each element in $V_l$ with probability at most $1/\rme$ for each $l \in [k]$ from Lemma \ref{lem:classical}, we have $\Pr(v \in S_k \mid V_1,\cdots,V_k,U) \le 1 / \rme$ for all $v \in V$.
	Let $g(S_k) = f(S_k \cup U) - f(U)$.
	By applying Lemma \ref{lem:useful2} to $g$, we have
    \begin{align}
        \bbE[f(S_k \cup U) | V_1,\cdots,V_k,U]
		&= \bbE[g(S_k) | V_1,\cdots,V_k,U]\nonumber\\
		&\ge (1 - 1 / \rme) g(\emptyset)\nonumber\\
		&= (1 - 1 / \rme) f(U).\label{eq:cup}
    \end{align}
	We consider the marginal gain of the algorithm obtained in the $l$th segment by fixing $S_{l-1}$.
	If $U \cap V_l \neq \emptyset$, it is a singleton, and let $u_l$ be an element in $U \cap V_l$.
	If $U \cap V_l = \emptyset$, let $u_l$ be a dummy element that does not affect the objective value.
	Define $s_l$ similarly as an element in $S_k \cap V_l$ or a dummy element.
	Since \Cref{alg:classical} is $1/\rme$-competitive and elements with negative marginal gain are not selected, we have
	\begin{equation}
		\bbE[ f(s_l | S_{l-1}) | V_1,\cdots,V_k,U,S_{l-1}] \ge \frac{1}{\rme} \max \{ f(u_l | S_{l-1}) , 0 \}.\label{eq:secretary}
    \end{equation}

	Let $U_l = U \cap \left(\bigcup_{j=1}^l V_j\right)$.
    Combining the above inequalities, we obtain
    \begin{align}
        \bbE[ f(S_k) \mid V_1, \cdots, V_k, U ]
		&= \sum_{l = 1}^k \bbE[f(s_l | S_{l-1}) \mid V_1, \cdots, V_k, U ]\nonumber \\
        &= \sum_{l = 1}^k \bbE[ \bbE[f(s_l | S_{l-1}) \mid V_1, \cdots, V_k, U, S_{l-1}] \mid V_1, \cdots, V_k, U]\nonumber \\
		&\ge \sum_{l = 1}^k \bbE\left[ \frac{1}{\rme} \max \{ f(u_l | S_{l-1}) , 0 \} \middle| V_1,\cdots, V_k, U \right]\nonumber \tag{due to \eqref{eq:secretary}}\\
		&\ge \sum_{l = 1}^k \bbE\left[ \frac{1}{\rme} f(u_l | S_k \cup U_{l-1}) \middle| V_1,\cdots, V_k, U \right] \tag{due to the submodularity} \nonumber \\
		&= \bbE\left[ \frac{1}{\rme} \left\{ f(S_k \cup U) - f(S_k) \right\} \middle| V_1,\cdots, V_k, U \right].\nonumber
    \end{align}
	By a simple calculation, we have
    \begin{align}
        \bbE[ f(S_k) \mid V_1, \cdots, V_k, U ]
		&= \bbE\left[ \frac{1}{1+\rme} f(S_k \cup U) \middle| V_1,\cdots, V_k, U \right]\nonumber\\
		&\ge \frac{1}{1+\rme} \left(1 - \frac{1}{\rme}\right) f(U).\nonumber \tag{due to \eqref{eq:cup}}
    \end{align}
	By taking the expectation about $V_1,\cdots,V_k$ and $U$ and substituting \eqref{eq:u}, we obtain
    \begin{equation*}
        \bbE[ f(S_k) ]
        \ge \frac{1}{1 + \rme} \left(1 - \frac{1}{\rme} \right)^2 f(S^*).
    \end{equation*}
\end{proof}

\section*{Acknowledgement}
This study is supported by JSPS KAKENHI Grant Number JP 18J12405.

\bibliographystyle{plainnat}
\bibliography{main}

%
%
%

\end{document}